\documentclass[10pt]{article}

\usepackage[margin = 1.0in]{geometry}   
\usepackage{amssymb}
\usepackage{amsmath}
\usepackage{palatino}
\usepackage{graphicx}
\usepackage[english]{babel}
\usepackage{amscd}
\usepackage{amsthm}
\usepackage{color}
\usepackage{enumerate}
\usepackage{hyperref}
\newtheorem{theorem}{Theorem}

\newtheorem{lemma}{Lemma}

\newtheorem{remark}{Remark}

\newcommand{\Prb}{\mathbb{P}}
\newcommand{\Exp}{\mathbb{E}}
\newcommand{\mN}{\mathcal{N}}
\newcommand{\mC}{\mathcal{C}}
\newcommand{\mA}{\mathcal{A}}

\title{A note on the network coloring game: A randomized distributed  $(\Delta +1)$-coloring algorithm\thanks{This research was supported by the Hellenic Foundation for Research and Innovation (H.F.R.I.) under the “1st Call for H.F.R.I. Research Projects to support Faculty members and Researchers and the procurement of high-cost research equipment grant” (Project Number: HFRI-FM17-2436).}}

\author{
Nikolaos Fryganiotis\thanks{School of Electrical and Computer Engineering, 
National Technical University of Athens, 
Zografou, Greece, 15780, e-mail: nikolasfryganiotis@yahoo.com}
\and
Symeon Papavassiliou\thanks{School of Electrical and Computer Engineering, 
National Technical University of Athens, 
Zografou, Greece, 15780, e-mail: papavass@mail.ntua.gr}
\and
Christos Pelekis\thanks{School of Electrical and Computer Engineering, 
National Technical University of Athens, 
Zografou, Greece, 15780, e-mail: pelekis.chr@gmail.com}
}

\begin{document}

\maketitle

\begin{abstract}  
The network coloring game has been proposed in the 
literature of social sciences as a model for  conflict-resolution circumstances.  
The players of the game are the vertices of a 
graph with $n$ vertices and maximum degree $\Delta$. 
The game is played over rounds, and in each round all players simultaneously choose a color from a set  
of available colors. Players have local information 
of the graph: they only observe the colors chosen by their neighbors and do not communicate or cooperate with one another. A player is 
happy when she has chosen a color that is 
different from the colors chosen by her 
neighbors, otherwise she is unhappy, and a configuration of colors for which all players are happy is a proper coloring of the graph. 
It has been shown in the literature that, when the players adopt a particular  
greedy randomized strategy, the game reaches a proper coloring of the graph within $O(\log(n))$ rounds, with high probability, provided the number of  colors available to each player is at least $\Delta+2$. In this note we show that a modification of the aforementioned greedy strategy yields likewise a proper coloring of the graph, provided the number of colors available to each player is at least $\Delta+1$, and results in a simple randomized distributed algorithm for the $(\Delta+1)$-coloring problem. 
\end{abstract}

\noindent{
\emph{Keywords}: graph coloring; games on graphs;
symmetric strategies; greedy algorithms; distributed algorithms
}

\noindent{
\emph{MSC (2010)}: 05C15; 05C57 
}

\section{Related work and main result}

\subsection{Notation}

Throughout the text, given a positive integer $n$, we denote by $[n]$ the set 
$\{1,\ldots,n\}$ and, given a finite set $F$, we denote by $|F|$ its cardinality. 
Given a  graph    $G=(V,E)$ and a vertex $v\in V$, 
we denote the \emph{neighborhood} of $v$ by $\mN(v)=\{u\in V : (u,v)\in E\}$. The cardinality of $\mN(v)$ is the \emph{degree} of $v$. All graphs considered in this note are finite, without loops, undirected and simple. 
A \emph{coloring} of a graph $G=(V,E)$ is an assignment of colors to the elements of $V$. 
A coloring which uses at most $k$ colors is called a \emph{$k$-coloring}. A $k$-coloring is \emph{proper} when no two adjacent vertices 
receive the same color.

\subsection{Related work}
 
The problem of 
determining the  smallest possible value of $k$ for which a graph admits a proper $k$-coloring 
is classical. It is well known to be NP-hard, and has attracted an immense attention both from a theoretical as well as a  practical perspective.  
A particular line of research focuses on  
games / algorithms for graph coloring (see, for example, \cite{Jensen_Toft, Lewis, Molloy_Reed} for textbooks devoting whole chapters to the topic). It is well-known that one can properly color a graph using $\Delta+1$ colors in linear time via a centralized algorithm. However the problem becomes more delicate when the algorithm is required to be distributed, a fact that gives  rise to the so-called distributed coloring problem.

Let $G$ be a graph on $n$ vertices and maximum degree $\Delta$.  
The \emph{distributed $k$-coloring problem} on $G$ is the problem of properly coloring the vertices of $G$, in a distributed manner, using a given set of $k$ available colors. The problem originates 
in a rather influential paper of Linial~\cite{Linial} and has attracted a great deal of attention. 
The above-mentioned centralized algorithm implies that at most $\Delta+1$ colors 
are enough to properly color $G$, and therefore  most of the work on distributed coloring naturally focuses on the $(\Delta+1)$-coloring problem. 
The corresponding literature is vast and is roughly divided in two lines of research: randomized distributed algorithms and deterministic distributed algorithms. 
In this article we shall be concerned with the former setting, and we refer the reader to~\cite{HKNT} for a recent account on the history and state of the art of the latter. 

Perhaps the first, simple, randomized distributed algorithm for the $(\Delta+1)$-coloring problem is due to  
Luby~\cite{Luby} (see also~\cite[Section~10.1]{Barenboim_Elkin}),
which finds a proper coloring of $G$ within $O(\log(n))$ rounds. Comparably simple randomized distributed algorithms, with $O(\log(n))$ running time, for the  $(\Delta+1)$-coloring problem have been reported by Collet, Fraigniaud and Penna~\cite{CFP}, Finocchi, Pasconesi and  Silvestri~\cite{FPS} and Johansson~\cite{Johansson}. 
Let us remark that  all of the aforementioned algorithms rely on a standing assumption of the distributed coloring problem which states that the graph $G$ is at the same time a communication network; an  assumption which allows the vertices to exchange messages over the edges of $G$.
In particular, the randomized algorithms  in~\cite{CFP, FPS, Johansson, Luby} require that at each round each vertex knows that status of her  neighbours, which means that each vertex communicates to her neighbours  whether she has any conflicts or not. The so-called \textsc{Local} model allows for adjacent vertices to exchange messages of arbitrarily large size, while the \textsc{Congest} model allows for adjacent vertices to exchange messages of $O(\log(n))$ bits per round. 
Given that there are several simple $O(\log(n))$-rounds randomized distributed algorithms for the $(\Delta+1)$-coloring problem, much of the current research focuses on minimizing the number of rounds until convergence, with the current best randomized $(\Delta+1)$-coloring algorithm in the \textsc{Local} model being due to Chang, Li, and
Pettie~\cite{CLP}, which runs in $O(\log^3(\log(n)))$ rounds.  

A basic idea in several of the above-mentioned distributed randomized algorithms is that  vertices with conflicts should utilize the set of colors that are not chosen by their neighbours.  
A bit more concretely, suppose that $v$ is a vertex having a conflict after a certain round, say $t$. This means that there is some other vertex, say $w$, 
that has chosen the same color as $v$. 
Now, for every vertex in $u\in G$ having a conflict, one can associate the set of ``available" colors, say $\mA_t(u)$, consisting of the colors that are not chosen by her neighbours after round $t$. 
Observe that, when $k\ge \Delta+1$, it holds $|\mA_t(u)|\ge 1$, for every vertex $u\in G$ having a conflict.   
Now, when $k=\Delta+1$,  it could happen that  $|\mA_t(v)|=|\mA_t(w)|=1$ and $\mA_t(v)=\mA_t(w)$. 
A worst-case scenario is an instance for which all  vertices in $\mN(v)\setminus \{w\}$ have no conflict, and have taken all remaining  colors. In such an instance, the vertices $v,w$ have to decide about how to choose their colors in the next round. 
The ``greedy approach" will result in an infinite loop, and one approach for handling such instances is to allow cooperation / communication among vertices.  
To the best of our knowledge, the first purely game-theoretic distributed randomized algorithm  that requires no communication or cooperation among vertices is due to Chaudhuri, Chung and Jamall~\cite{CCJ}. The main idea behind the algorithm in~\cite{CCJ} is that when $k\ge \Delta + 2$, then it holds $|\mA_t(u)|\ge 2$, for every vertex $u\in G$ having a conflict, and such a player can choose, in the next round, a color  uniformly at random from the set $\mA_t(u)$. 
A vertex with no conflict just sticks to her choice in all subsequent rounds. The main result from~\cite{CCJ} states that this algorithm converges in $O(\log(n))$ rounds to a proper coloring of the graph. 
In this article we extend the idea from~\cite{CCJ}. Our observation is that if we  include the color chosen by vertex $u$ after round $t$
to the set $\mA_t(u)$, then vertex $u$ has at least two available colors to choose from in the next round; a fact that holds true even in the case $k=\Delta+1$. 
This allows to resolve potential conflicts among vertices and results in a distributed randomized algorithm for the $(\Delta+1)$-coloring problem that requires no exchange of information among vertices. Our algorithm is stated in terms of a symmetric strategy among players (i.e., the Frugal strategy defined below) of a particular game played on a graph, which we now formally define.

\subsection{A network coloring game: Main result}

We shall be concerned with a 
particular game for graph coloring, which is 
referred to as 
the \emph{network coloring game}. 
The network coloring game is introduced and studied 
empirically in~\cite{Kearns_et_al} as a model for conflict-resolution circumstances. The game is  
played on a graph 
$G=(V,E)$ with $|V|=n$ vertices and maximum degree $\Delta$. 
Each vertex of $G$ is a player, having a set
of $k$ available colors, and participates 
in a game that is played over a number of  rounds. 
In each round all players simultaneously choose a color from their set of available colors, 
which is assumed to be the set $[k]$. 
Players have local information on the graph: they can only observe the colors of their neighbors, and cannot communicate or cooperate with one another.  
A player is \emph{happy} if she 
has chosen a color that is different from the colors
chosen by her neighbors. Otherwise, the player is \emph{unhappy}. In game-theoretic language, the payoff to a player is $1$, if she is happy, and is $0$ if she is unhappy, 
and a configuration of colors  for which every player 
receives payoff $1$ is a \emph{Nash equilibrium} of the game, in the sense that no player has the incentive to change 
strategy under such a configuration. 

The problem is to find a symmetric strategy for the players (i.e., a strategy that 
is the same for all players) 
that achieves  
convergence to a Nash equilibrium after  
a finite number of rounds, using the smallest possible value for $k$.                  
Proving that a particular symmetric strategy 
is optimal (i.e., it minimizes the (expected) time
to equilibrium) 
is probably an elusive problem, and a more realistic line of research  
is to propose ``reasonable" symmetric  strategies and estimate their time to Nash  equilibrium.

Such a strategy has been proposed in~\cite{CCJ}, 
and will be referred to as the \emph{Greedy strategy}. 
In order to formally define the Greedy strategy, 
let us introduce some  notation. 
Let $c_t(v)$ be the color chosen by player $v$ after round $t$. Notice that a player is 
unhappy after round $t$ if there exists $u\in\mN(v)$ such that $c_t(u)=c_t(v)$. Moreover, let 
$\mC_t(v)$ be the set consisting of the colors 
chosen by the neighbors of $v$ after round $t$, i.e., $\mC_t(v)=\bigcup_{u\in\mN(v)}\{c_t(u)\}$.

\textbf{Greedy strategy}. Suppose that $k\ge \Delta+2$ and that each player adopts the following strategy: if a player, say $v$, is  happy after a certain round, say $t$, then she sticks to her color for all subsequent rounds, i.e., $c_s(v)=c_t(v)$ for all $s>t$. 
If she is unhappy after round $t$, then in the next round she \emph{changes} color, and chooses the color $c_{t+1}(v)$ 
uniformly at random from the set $[k]\setminus \mC_t(v)$, consisting  of those  colors which are not chosen by any of her neighbors after round $t$.

\begin{remark}
Notice that, when all players in the network coloring game adopt the Greedy strategy, 
a happy player remains happy in all subsequent rounds. 
Notice also that, since $k\ge \Delta+2$, it holds 
$|[k]\setminus \mC_t(v)|\ge 2$ for all $v\in V$ and 
all rounds $t\ge 1$. In particular, this implies that for every unhappy player, there are always at least two colors that are not chosen by her neighbors. Let us also remark that the assumption $k\ge \Delta+2$ is crucial in the Greedy strategy. Indeed, if $k=\Delta+1$ then the Greedy strategy may result in a game that never reaches a Nash equilibrium, as can be easily seen by employing the strategy to a triangle-graph (see also~\cite[Theorem~2]{CCJ}). 
\end{remark}

It is shown in~\cite{CCJ} that, when all players adopt the Greedy strategy, the expected number  of unhappy players decays exponentially in
each round. More concretely, for every player 
$v\in V$, let $\tau_v$ denote the first round after which player 
$v$ becomes happy. Then $\tau=\max_v \tau_v$ is the first round after which all players are happy. 
Notice that at round $\tau$ the graph is properly $k$-colored and the payoff to every player equals $1$. 
In other words, after time $\tau$ the game reaches a Nash 
equilibrium. 
The following statement is the main result in~\cite{CCJ}. 

\begin{theorem}[Chaudhuri, Chung-Graham, Jamall~\cite{CCJ}]
\label{ccj}
Let $G$ be a graph on $n$ vertices and maximum degree $\Delta$. Suppose that the
number of available colors to each player is $k\ge \Delta+2$ and that each player in the network coloring game adopts the Greedy strategy. Let $\tau$ be the first round after 
which all players are happy. 
Then, for any starting assignment of colors to the vertices, it holds 
\[
\Prb\left(\tau \le C \cdot \log\left(\frac{n}{\delta}\right)\right) \ge 1-\delta\, ,
\]
where $C=1050e^9$ and $\delta>0$ is arbitrarily small. 
\end{theorem}

In other words, when the players in the network coloring game adopt the Greedy strategy, the game converges to 
a Nash equilibrium in $O\left(\log\left(\frac{n}{\delta}\right)\right)$ rounds with  probability at least $1-\delta$.  
Improved estimates on the expected value of $\tau$ 
can be found in~\cite[Theorem~3]{PelSch}. 
In this note we combine some ideas from the approaches in~\cite{CCJ} and~\cite{PelSch} and deduce a refinement of Theorem~\ref{ccj}. A bit more concretely, we show that a  modification 
of the above-mentioned Greedy strategy allows to 
deduce the same conclusion as in Theorem~\ref{ccj},  
subject to the assumption that the number of colors available to each player is at least $\Delta+1$. 
We refer to this modified Greedy strategy as a 
\emph{Frugal strategy}, which is formally defined as 
follows.

\textbf{Frugal strategy}. Suppose that $k\ge \Delta+1$ and that each player in the network coloring game first chooses, independently of all other choices, a color uniformly at random from the set $[k]$, and then  adopts the following strategy:
if a player, say $v$, is  happy after round $t\ge 1$, then she sticks to her color for all subsequent rounds, i.e., $c_s(v)=c_t(v)$ for all $s>t$. 
If she is unhappy after round $t$, then in the next round 
she chooses uniformly at random a color from the set 
$\{c_t(v)\} \cup ([k]\setminus \mC_t(v))$. 

In other words, under the Frugal strategy, a player who is unhappy after round $t\ge 1$ chooses in the next round a color uniformly at 
random from the set consisting of her color-choice after round $t$ and the set of colors that are not chosen by her neighbors after round $t$. 

\begin{remark}
Notice that, since $k\ge \Delta+1$, it holds 
$|\{c_t(v)\} \cup([k]\setminus \mC_t(v))|\ge 2$ for all $v\in V$ and 
all rounds $t\ge 1$. In particular, every unhappy player
has at least two colors to choose from in the next round. Notice also that, in contrast to the Greedy strategy, under the Frugal strategy 
an unhappy player may not change color in the next round. 
\end{remark}

In this note we imitate the analysis of the proof of Theorem~\ref{ccj}, from~\cite{CCJ},  
and show that the Frugal strategy converges to a 
Nash equilibrium in a finite number of rounds. 
More precisely, we obtain the following  refinement of Theorem~\ref{ccj}. 
Recall (see~\cite{Shaked_Shanthikumar}) that a random variable $X$ 
is \emph{stochastically smaller} than 
a random variable $Y$, denoted 
$X\le_{st} Y$, if it holds   
$\Prb(X> t) \le \Prb(Y> t)$, for all $t$.

\begin{theorem}
\label{main:thm}
Let $G$ be a graph on $n$ vertices and maximum degree $\Delta$. Suppose that the
number of  available colors to each player is $k\ge \Delta+1$ and that each player in the network coloring game adopts the Frugal strategy. Let $\tau$ be the first round after 
which all players are happy.
Then $\tau$ is stochastically smaller than a random variable $T$ such that
\[
\Exp(T) \le  \frac{2}{\mu} (1 + \log(n))  \quad \text{ and } \quad  \text{Var}(T)\le \frac{4n}{\mu^2} \, , 
\]
where $\mu = -\log\left(1-\frac{1}{2^6e^5}\right)\approx 0.000105$.  
\end{theorem}

We prove Theorem~\ref{main:thm} in the next section. 
The proof mimics the proof of Theorem~\ref{ccj}, from~\cite{CCJ}, and proceeds in two steps. 
The first step concerns a lower estimate on the probability that a player, who 
is unhappy after a certain round, say $t$, has ``enough" 
available colors after round $t+1$. The second step concerns a lower estimate on the probability 
that the player becomes happy after round $t+2$, given that 
she has ``enough" available colors after round $t+1$. 
Both estimates do not depend on $\Delta$  
and, when combined, yield a lower bound on the probability that an 
unhappy player becomes happy after two rounds, which is also independent of $\Delta$. 
Let us remark that the analysis over two rounds is crucial; over a single round it could happen that an unhappy player becomes happy in the next round with probability $\frac{1}{2^{\Delta}}$, an estimate which clearly depends on $\Delta$. 
The  proof of Theorem~\ref{main:thm} is completed  
using  ideas from the theory of maximally dependent random variables.

\section{Proof of main result}  

In this section we prove Theorem~\ref{main:thm}. We assume that each player in the network coloring game adopts the 
Frugal strategy. 
We begin with a lower estimate on the probability that 
a player, who is unhappy after  a certain round, receives 
``enough" available colors in the next round. 
In order to be more precise, we need some 
extra piece of notation that will remain fixed throughout this section. 

Recall that 
$c_t(v)$ is the color chosen by player $v$ after round $t$ and that $\mC_t(v)$ is 
the set of colors chosen by its neighbors. 
For every $t\ge 1$, 
let $H_t$ denote the set of happy players after 
round $t$, and set $U_t = V\setminus H_t$ be the set of unhappy players after round $t$.  
Given $v\in U_t$, let 
\[
\mA_t(v) = \{c_t(v)\} \cup ([k]\setminus \mC_t(v)) 
\]
be the set of colors \emph{available} to $v$ after round $t$; hence in the next round player $v$ chooses the color $c_{t+1}(v)$  
uniformly at random 
from the set $\mA_{t}(v)$. Let also $p_t(v)= \frac{1}{|\mA_t(v)|}$ denote the probability 
with which the unhappy player $v$ chooses her color 
in the next round. For $v\in H_t$, set $\mA_t(v) = \{c_t(v)\}$ and $p_t(v)=1$. 
Similarly, given a vertex $v\in V$, let $H_t(v)$ denote the set 
of happy neighbors of $v$ after round $t$, and 
let 
\begin{equation*}%\label{f_set}
    F_t(v) = \bigcup_{u\in H_t(v)} \{c_t(u)\}
\end{equation*}
be the set of colors chosen by the happy neighbors of $v$ after round $t$, and 
$U_t(v) = \mN(v)\setminus H_t(v)$ be the set of unhappy neighbors of $v$. 
Notice that every color in the set $[k]\setminus F_t(v)$ has a non-zero chance of not being chosen by any unhappy neighbour of $v$, and so
has a non-zero chance of 
belonging to the set 
$\mA_{t+1}(v)$. 
Finally, let $f_t(v) = |F_t(v)|$, 
and notice that, since happy players stick to their choice, the sequence $\{f_t(v)\}_{t\ge 1}$ is  non-decreasing. 
Thus the number of colors available to $v$ after round $t + 1$ as well as after round $t+2$ is at most $k-f_t(v)$. 
We now establish a lower bound on the probability that the number of colors available to player $v\in U_t$ after round $t + 1$ is at least $\frac{k-f_t(v)}{5}$.

\begin{lemma}\label{first_round}
For each $t\ge 1$ and each $v\in U_t$, it holds 
\[
\Prb\left(|\mA_{t+1}(v)|\ge \frac{k-f_t(v)}{5}\right)\ge \frac{1}{2^4} \, . 
\]
\end{lemma}
\begin{proof} 
To simplify notation, let $f:=f_t(v)$.
We first estimate $\Exp(|\mA_{t+1}(v)|)$ from below; the result then will follow 
from Markov's inequality. 
Recall that every color from the 
set $[k]\setminus F_t(v)$ has a positive chance of being an element of   $\mA_{t+1}(v)$. 
The probability that a fixed color $i\in [k]\setminus F_t(v)$ is not chosen by any $u\in U_t(v)$ in the next round is equal to
\[
\prod_{\{u\in U_t(v): i\in \mA_t(u)\}} (1-p_t(u))  \, .
\]
Therefore, using the arithmetic-geometric means inequality, we have  
\begin{eqnarray*}
\Exp(|\mA_{t+1}(v)|) &\ge& \sum_{i\in [k]\setminus F_t(v)}\, \prod_{\{u\in U_t(v): i\in \mA_t(u)\}} (1-p_t(u))  \\ 
&\ge& (k-f)\cdot \left(\prod_{i\in [k]\setminus F_t(v)} \prod_{\{u\in U_t(v): i\in \mA_t(u)\}} (1-p_t(u)) \right)^{\frac{1}{k-f}} \\
&\ge& (k-f)\cdot \left(\prod_{u\in U_t(v)}\, \prod_{i\in \mA_t(u)} (1-p_t(u)) \right)^{\frac{1}{k-f}} \\
&=&  (k-f) \cdot \left( \prod_{u\in U_t(v) }\, \, \left(1- \frac{1}{|\mA_t(u)|}\right)^{|\mA_t(u)|}
\right)^{\frac{1}{k-f}} \, .
\end{eqnarray*}
Now notice that for every $u\in U_t(v)$ it holds $|\mA_t(u)|\ge 2$; hence  
$1- \frac{1}{|\mA_t(u)|}>0$. 
Since the sequence $\{(1-1/m)^m\}_{m\ge 2}$ is non-decreasing and $|\mA_t(u)|\ge 2$, it follows that 
$\left(1- \frac{1}{|\mA_t(u)|}\right)^{|\mA_t(u)|} \ge \left(1-\frac{1}{2}\right)^2 = \frac{1}{4}$. 
Putting the above together, we conclude 
\[
\Exp(|\mA_{t+1}(v)|) \ge  (k-f)\cdot \left(\frac{1}{4}\right)^{\frac{|U_t(v)|}{k-f}} \, .
\]
Now, since $k\ge \Delta+1$, it holds $|U_t(v)|\le \Delta -f \le k-1-f$, and thus $\frac{|U_t(v)|}{k-f}\le 1$. This implies that $\left(\frac{1}{4}\right)^{\frac{|U_t(v)|}{k-f}}\ge \frac{1}{4}$ and therefore  
\[
\Exp(|\mA_{t+1}(v)|) \ge  \frac{k-f}{4} \, . 
\]
To complete the proof, let $X=k-f-|\mA_{t+1}(v)|$ and apply the lower estimate on $\Exp(|\mA_{t+1}(v)|)$, together with Markov's inequality, to deduce 
$\Prb\left(|\mA_{t+1}(v)| < \frac{k-f}{5}\right) =
\Prb\left(X >\frac{4(k-f)}{5}\right) < 
\frac{5\cdot\Exp(X)}{4(k-f)} \le  \frac{15}{16}$,
as desired. 
\end{proof}

In the next lemma we estimate from below the probability that a player, who is unhappy after round $t$, becomes happy after two rounds. 
This will require some additional notation. 

Fix a player $v\in U_{t+1}$. 
Since $v$ is unhappy, it follows that there 
exists $u\in\mN(v)$ such that $c_{t+1}(u)=c_{t+1}(v)$. 
There are two kinds of unhappy neighbors of $v$
participating in the game. Those that have the same color as player $v$, and those that have different color.
This partitions the set $U_{t+1}(v)$ into the sets  
\[
S_{t+1}(v) = \{u\in U_{t+1}(v): c_{t+1}(u)=c_{t+1}(v)\} \quad  \text{ and }  \quad D_{t+1}(v) = \{u\in U_{t+1}(v): c_{t+1}(u)\neq c_{t+1}(v)\} \, .
\]
Observe that for every $u\in S_{t+1}(v)$ 
it holds $c_{t+1}(v)\in \mA_{t+1}(u)$, while 
for every $u\in D_{t+1}(v)$ it holds 
$c_{t+1}(v)\notin \mA_{t+1}(u)$. 
Moreover, since $c_{t+1}(u)\in\mA_{t+1}(u)$, it  follows that for every $u\in D_{t+1}(v)$ 
the set $\mA_{t+1}(u)$ contains a color, namely, color $c_{t+1}(u)$, that is \emph{not} contained in $\mA_{t+1}(v)$.

\begin{lemma}\label{second_round}
It holds 
\[
\Prb(v\in H_{t+2} \, | \,  v\in U_t) \ge \frac{1}{2^6 e^5} \, .
\]
\end{lemma}
\begin{proof} 
Let $v\in U_t$. Then, conditional on 
$\mA_{t+1}(v)$ and $v\in U_{t+1}$, 
the probability that player $v$ is happy 
after round $t+2$ is the average of the probabilities that a fixed color from $\mA_{t+1}(v)$ is not chosen by any unhappy  neighbor $u\in U_{t+1}(v)$.  
To simplify notation, let us define, for each 
color $i\in [k]$, the sets 
\[
S_{t+1}^{(i)}:= \{u\in S_{t+1}(v)\,:\, i\in\mA_{t+1}(u)\} \quad \text{ and } \quad 
D_{t+1}^{(i)}:= \{u\in D_{t+1}(v)\,:\, i\in\mA_{t+1}(u)\} \, .
\]
For $i\in \mA_{t+1}(v)$ and 
$u\in U_{t+1}(v)$, let $q(i;u):=\Prb(c_{t+2}(u) \neq i)$ be the probability that player $u$ does not choose 
color $i$ in the next round. Then 
the probability that a fixed color 
$i\in\mA_{t+1}(v)$ is not chosen by any 
player $u\in U_{t+1}(v)$ in the next round is equal to 
\[
\prod_{u\in S_{t+1}^{(i)}} q(i; u) \,
\prod_{u\in D_{t+1}^{(i)}} q(i; u) \, .
\]
Hence, conditional on $\mA_{t+1}(v)$ and $v\in U_{t+1}$, the probability 
that player $v$ is happy after round $t+2$ equals  
\begin{eqnarray*}
\pi_{t+2} &:=& \frac{1}{|\mA_{t+1}(v)|} \, 
\sum_{i\in \mA_{t+1}(v)}\, \prod_{u\in S_{t+1}^{(i)}} q(i;u)  
\prod_{u\in D_{t+1}^{(i)}} q(i;u) \\
&\ge&  \left( \prod_{i\in \mA_{t+1}(v)}\,\, \prod_{u\in S_{t+1}^{(i)}} q(i;u) 
\prod_{u\in D_{t+1}^{(i)}} q(i;u) \right)^{1/|\mA_{t+1}(v)|} \\
&=&  \left( \prod_{i\in \mA_{t+1}(v)}\,\, \prod_{u\in S_{t+1}^{(i)}} q(i;u) 
 \right)^{1/|\mA_{t+1}(v)|}  \left( \prod_{i\in \mA_{t+1}(v)}\,\, 
\prod_{u\in D_{t+1}^{(i)}} q(i;u) \right)^{1/|\mA_{t+1}(v)|}
\end{eqnarray*}
where the estimate follows from the arithmetic-geometric means inequality. 
Now observe that $c_{t+1}(v)\in \mA_{t+1}(v)$, 
and therefore it holds $q(c_{t+1}(v);u)=1$, for $u\in D_{t+1}(v)$.  
Moreover, for each $u\in D_{t+1}(v)$, the set $\mA_{t+1}(u)$ contains at least one color (namely, $c_{t+1}(u)$) that does not belong to $\mA_{t+1}(v)$. 
The last two observations imply that 
\begin{eqnarray*}
\prod_{i\in\mA_{t+1}(v)}\,\,\prod_{u\in D_{t+1}^{(i)}}  q(i;u) &=& 
\prod_{i\in\mA_{t+1}(v)\setminus \{c_{t+1}(v)\}}\,\,\prod_{u\in D_{t+1}^{(i)}}  q(i;u) \\ 
&\ge& \prod_{u\in D_{t+1}(v)} \,\, \prod_{i\in\mA_{t+1}(u)\setminus \{c_{t+1}(u)\}} \left(1-\frac{1}{|\mA_{t+1}(u)|} \right) \\
&=& \prod_{u\in D_{t+1}(v)} \,\left(1-\frac{1}{|\mA_{t+1}(u)|} \right)^{|\mA_{t+1}(u)|-1} \\
&\ge& \left(\frac{1}{e}\right)^{|D_{t+1}(v)|} \, , 
\end{eqnarray*}
where the last estimate follows from the fact that $\left(1-\frac{1}{m}\right)^{m-1}\ge \frac{1}{e}$, when $m\ge 2$. 
Similarly, we have 
\begin{eqnarray*}
\prod_{i\in \mA_{t+1}(v)}\,\, \prod_{u\in S_{t+1}^{(i)}} q(i;u) &\ge& \prod_{u\in S_{t+1}(v)} \, \prod_{i\in\mA_{t+1}(u)} \left(1-\frac{1}{|\mA_{t+1}(u)|}\right) \\
&=& \prod_{u\in S_{t+1}(v)} \left(1-\frac{1}{|\mA_{t+1}(u)|}\right)^{|\mA_{t+1}(u)|} \\
&\ge& \left(\frac{1}{4} \right)^{|S_{t+1}(v)|} \, .
\end{eqnarray*}
Putting the above together, we conclude that 
\[
\pi_{t+2} \ge \left(\frac{1}{4} \right)^{\frac{|S_{t+1}(v)|}{|\mA_{t+1}(v)|}}\cdot\left(\frac{1}{e}\right)^{\frac{|D_{t+1}(v)|}{|\mA_{t+1}(v)|}} \, .
\]
From Lemma~\ref{first_round} we know that with probability at least $1/2^4$ it holds
$|\mA_{t+1}(v)|\ge \frac{k-f_t(v)}{5}\ge \frac{k-|H_t(v)|}{5} \ge \frac{k-|H_{t+1}(v)|}{5}$. 
Furthermore, observe that  
$|D_{t+1}(v)|\ge k-|H_{t+1}(v)| - |\mA_{t+1}(v)|$. Since 
\[
|S_{t+1}(v)|+ |D_{t+1}(v)|=|U_{t+1}(v)|\le \Delta -|H_{t+1}(v)|\le k-1-|H_{t+1}(v)|\le k-|H_{t+1}(v)|\, ,
\]
it follows that 
$|S_{t+1}|\le |\mA_{t+1}(v)|$; 
hence $\pi_{t+2} \ge \frac{1}{4}\cdot \left(\frac{1}{e}\right)^{\frac{|D_{t+1}(v)|}{|\mA_{t+1}(v)|}}$. 
Since it clearly holds $|D_{t+1}(v)|\le k-|H_{t+1}(v)|$ we conclude that, 
conditional on the event that 
$|\mA_{t+1}(v)|\ge  \frac{k-|H_{t+1}(v)|}{5}$,
it holds $\pi_{t+2}\ge \frac{1}{4e^5}$. 
The result follows. 
\end{proof}

We now turn into the proof of our main result. Given $v\in V$, let $\tau_v$ be the 
first round after which player $v$ is happy 
and set $\tau = \max_v \tau_v$. 
We want to upper bound the expected value 
of $\tau$. Notice that the random variables 
$\tau_v, v\in V$, are \emph{not} mutually independent and our bound on $\tau$ 
will be a worst-case estimate. 
To this end, we follow~\cite{PelSch} and  employ ideas from 
the theory of maximally dependent random variables.
Given a real number $\mu>0$, 
let $E_{\mu}$ denote an exponential random variable of parameter $\mu$.

\begin{lemma}\label{third_lemma}
For every $v\in V$, it holds $\tau_v \le_{st} 2\cdot E_{\mu}$, where $\mu =-\log\left(1-\frac{1}{2^6e^5}\right)$. 
\end{lemma}
\begin{proof} 
We have to show that $\Prb(\tau_v > t) \le \Prb(E_{\mu} > \frac{t}{2})$, for all $t$. Notice that 
\[
\Prb(\tau_v >1) = 1 - \left(1-\frac{1}{k}\right)^{|\mN(v)|} \le 
1 - \left(1-\frac{1}{k}\right)^{k-1} \le 
1 - \frac{1}{e} \le 1- \frac{1}{2^6e^5} \, . 
\]
From 
Lemma~\ref{second_round} we know that  $\Prb(\tau_v > t+2\, |\, \tau_v>t)=\Prb(v\in U_{t+2}\, |\, v\in U_t)\le 1- \frac{1}{2^6e^5}$ holds true for 
every $t\ge 1$. 
Now notice that when $t$ is odd, say $t=2m+1$, it holds 
\begin{eqnarray*}
\Prb(\tau_v > t) &=& \Prb(\tau_v >1)\cdot\prod_{i=1}^{m}\Prb(\tau_v>2i+1\, |\, \tau_v>2i-1) \\
&\le& \left( 1- \frac{1}{2^6e^5} \right)^{m+1}\\
&\le& \left( 1- \frac{1}{2^6e^5} \right)^{t/2} \\
&=& \Prb\left(E_{\mu} > \frac{t}{2}\right) \, .
\end{eqnarray*}
If $t$ is even, the proof is similar, and the result follows. 
\end{proof}

Now let $\tau$ be the first round after which 
all players are happy. 
Then $\tau = \max_v\tau_v$. 
The proof of our main result is almost complete. 
Given two random variables $X,Y$, let $X\sim Y$ denote that fact that they have the same distribution.

\begin{proof}[Proof of Theorem~\ref{main:thm}]
Lemma~\ref{third_lemma} implies that, for 
all $v\in V$, it holds 
$\tau_v \le_{st} Y_v$, where $Y_v\sim 2\cdot E_{\mu}$.
Since $\tau_v \le_{st} Y_v$ it follows (see~\cite[Theorem~1.A.1]{Shaked_Shanthikumar}) that 
there exist random variables $\hat{\tau}_v, \hat{Y}_v$ 
such that $\hat{\tau}_v\sim\tau_v$, $\hat{Y}_v\sim Y_v$ and $\hat{\tau}_v \le \hat{Y}_v$ with probability $1$. 
Hence $\max_v \hat{\tau}_v \le \max_v \hat{Y}_v$ with probability $1$. Since $\tau\sim\max_v \hat{\tau}_v$, 
we conclude that $\tau\le_{st} 2\cdot M$, where $M$ is the maximum of $n$ exponential random variables, say $\{X_v\}_{v\in V}$, of parameter $\mu$. Hence $\Exp(\tau)\le 2\cdot\Exp(M)$ and 
it is therefore enough to establish an upper bound on  $\Exp(M)$. 
To this end, we borrow ideas from~\cite{Lai_Robbins}. 
Note that for every real number $a$ we have 
$M \le a + \sum_v \max\{X_v-a, 0\}$; hence 
\[
\Exp(M) \le a + \sum_v \Exp(\max\{X_v-a, 0\}) = a + n \int_{a}^{\infty} (1- F(x)) \, dx \, ,
\]
where $F(\cdot)$ is the distribution function of $X_v\sim E_{\mu}$. Let $h(a) = a + n \int_{a}^{\infty} (1- F(x)) \, dx$, defined for real $a$, and notice that 
$h(\cdot)$ attains its minimum at $a_n:= F^{-1}(1-\frac{1}{n})$. Since $F(x)=1-e^{\mu x}$, 
we deduce that 
\[
\Exp(M) \le a_n + n\int_{a_n}^{\infty} e^{-\mu x} dx 
= \frac{1}{\mu} (1 + \log(n)) \, ,
\]
as desired. Finally, the main result from~\cite{Rychlik} implies  
that $\text{Var}(M) \le n\cdot\text{Var}(E_{\mu})=\frac{n}{\mu^2}$. 
The result follows upon letting $T\sim 2\cdot M$. 
\end{proof}

\end{document}